\documentclass[pra,twocolumn, showpacs, showkeys, secnumarabic, aps, amsmath, amssymb, nofootinbib, superscriptaddress, longbibliography, floatfix, table-of-contents, eqsecnum, dblfloatfix]{revtex4-2}

\usepackage[pdftex]{graphicx}
\usepackage{mathrsfs}
\usepackage[colorlinks, breaklinks, urlcolor={blue}, linkcolor={red}, citecolor={blue}]{hyperref}
\usepackage{array}
\usepackage{amsmath}
\usepackage{amsthm}
\usepackage{type1cm}
\usepackage{lettrine}
\usepackage[english]{babel}
\usepackage{lmodern}
\usepackage{microtype}
\usepackage{enumitem}
\usepackage{booktabs}
\usepackage{pst-node}
\usepackage{tikz-cd} 
\usepackage[T1]{fontenc}
\usepackage[boxed, vlined]{algorithm2e}
\usepackage{braket}
\usepackage{xcolor}

\newtheoremstyle{nopunct}
  {}{}
  {\itshape}
  {}
  {\bfseries}
  {}
  { }
  {\thmname{#1}\thmnumber{ #2}\thmnote{ (#3)}}

\theoremstyle{nopunct}
\newtheorem{theorem}{Theorem}

\newtheorem{definition}{Definition}

\begin{document}

\title{Influence of Noninertial Dynamics on Static Quantum Resource Theories}

\author{Saveetha Harikrishnan}
\email[]{saveehari@gmail.com}
\affiliation{Department of Physics, Sri Sai Ram Engineering College,Chennai, India}

\author{Tim Byrnes}
\email[]{tim.byrnes@nyu.edu}
\affiliation{New York University Shanghai, NYU- ECNU Institute of Physics at NYU Shanghai, 567 West Yangsi Road, 200124,
China} 
\affiliation{State Key Laboratory of Precision Spectroscopy, School of Physical and Material Sciences, East China Normal
University, Shanghai 200062, China}
\affiliation{Center for Quantum and Topological Systems, NYUAD Research Institute, New York University Abu Dhabi, UAE} 
\affiliation{Department of Physics, New York University, New York, NY 10003, USA}

\author{Chandrashekar Radhakrishnan}
\email[]{chandrashekar10@gmail.com}
\affiliation{New York University Shanghai, NYU- ECNU Institute of Physics at NYU Shanghai, 567 West Yangsi Road, 200124, China} 

\homepage[]{https://sites.google.com/view/chandrashekar}

\begin{abstract}
The effect of noninertial dynamics on static quantum resource theories is investigated. To this end, we first show the equivalence between 
noninertial effects and a completely positive, trace-preserving (CPTP) map.  In this formulation, the Unruh effect is equivalent to a bosonic 
amplifier channel. The effect of this map on a generic quantum resource is investigated by studying the role of the CPTP map on the three core 
ingredients of a resource theory, namely, the free states, the free operations and the resource quantifiers. We show several general statements can be made about these three components of a resource theory in the presence of noninertial motion.  
\end{abstract}

\maketitle


\section{Introduction}

Quantum information theory has predominantly been developed under the assumption that quantum systems are either stationary or reside in inertial
frames of reference. However, in space-based applications \cite{yin2017satellite,lu2022micius} and otherwise, quantum systems may undergo 
accelerated or noninertial motion, giving rise to phenomena that 
cannot be captured within the inertial framework.  Several investigations
\cite{peres2004quantum,terno2006introduction,peres2003quantum,fuentes2010lecture,Nicolai2010inertial,alsing2003teleportation,grochowski2017effect,
friis2013relativistic,tjoa2022quantum,landulfo2016nonperturbative}
have been carried out to study the information theoretic aspects of quantum systems undergoing noninertial motion.  
Bipartite entanglement was the first type of quantum correlation to be investigated in a noninertial frame of reference 
\cite{fuentes2005alice}, where it was found that it vanishes in the infinite acceleration limit.  Meanwhile, tripartite entanglement 
though decreasing with acceleration was shown to be finite \cite{Hwang2001entanglement} even in the limit of infinite acceleration. 
Through the investigations on quantum discord \cite{datta2009quantum} it was also shown that the total quantum correlations also 
decrease but remain finite as the acceleration tends to infinity.  More recently, a comprehensive investigation of tripartite 
quantum coherence was carried out in Ref. \cite{harikrishnan2022accessible}, where it was found that  irrespective of its 
distribution coherence decreases with acceleration but remains finite under infinitely large acceleration. Numerous other studies investigating the effect of acceleration on properties such as 
entanglement 
\cite{fuentes2005alice,alsing2006entanglement,martin2011redistribution,bruschi2012particle,wang2011multipartite,Hwang2001entanglement,xiao2011mixed,
Adesso2007entanglement}, discord \cite{wang2010classical,brown2012vanishing,datta2009quantum,ramzan2014discord,sugumi2016discord,qiang2015geometric},
quantum steering \cite{Wang2016quantumsteering,wu2025gaussian}
and coherence \cite{wang2016irreversible,wu2020quantum,wu2021quantum,harikrishnan2022accessible} have been performed.  From these works 
it was found that generally quantum correlations decrease as a result of noninertial motion.  This can be understood  primarily due to the decohering influence of acceleration on the quantumness of the system.

Of the many quantum properties, nonlocality, steering, entanglement, discord and coherence form a subset that displays a hierarchical 
structure \cite{adesso2016measures,ma2019operational}  with respect to their occurrence in a quantum state.  
The hierarchical structure reflects the tendency of quantum properties to be 
lost in a specific order when a pure state is combined with a maximally mixed state.  A similar hierarchical loss of quantum properties is also 
observed in open quantum systems 
\cite{paulson2021hierarchy,chakrabarty2011study,bellomo2011dynamics,qasimi2011comparison,macieszczak2019coherence,ma2013multipartite,
cao2020fragility,radhakrishnan2019dynamics}.  
In open quantum systems, the hierarchical loss of quantum properties arises from the 
interaction between the system and the environment.  During this interaction, the system becomes entangled with the environmental degrees of freedom.  
Consequently tracing out the environment to obtain the reduced system state leads to a loss of quantum information.  This loss manifests as mixedness 
and results in decoherence leading to the decrease in quantum correlations in a hierarchical manner.

Quantum systems undergoing noninertial motion can be described using Rindler coordinates, where the quantum properties are distributed between two 
causally disconnected Rindler regions.  Since an observer in one region has no access to the degrees of freedom in the other, one of the two 
regions is traced out, resulting in a mixed state.  This mixedness leads to a reduction in the quantum resources of the system and thereby 
giving rise to a hierarchical decay of the resources analogous to the behavior  observed in open quantum systems. This loss of quantumness in 
noninertial motion can also be understood in terms of the Unruh effect \cite{Unruh1976}.  Here, an accelerated observer perceives the vacuum as a 
thermal bath of virtual particles.  The interaction of the quantum system with this virtual bath induces decoherence in the system, 
analogous to the effect of an environment in open quantum system.  In this framework, tracing over the inaccessible Rindler region is identical
to the tracing over the bath degrees of freedom and hence it similarly results 
in the mixedness of the system state.  We note that some previous studies have considered noninertial motion as a quantum channel.  However, its 
properties were not examined in depth \cite{aspachs2010optimal}, or alternative formulations were not in general trace preserving 
\cite{ahn2018unruh} (see Appendix \ref{app:ahn}).  For this reason, some further clarification of the approach seems to be necessary.

In this paper, we present the effects of noninertial motion on a generic quantum resource theory.  
We first formulate noninertial motion as a completely positive trace preserving (CPTP) map, which is a generic way to describe open quantum system 
dynamics \cite{schumacher1996sending}. The CPTP map that we construct can be employed for multilevel multipartite systems, generalizing the approach. 
The CPTP map provides an open system framework to understand noninertial dynamics.  This approach is firmly rooted 
in the Stinespring dilation theorem which is the underlying theoretical basis of an open quantum system. We then use these results to describe the 
influence of noninertial motion on a generic static quantum resource theory. We study the effects of the map on the three important components of a 
resource theory, namely the free states, the free operations and the resource quantifiers.

The manuscript is structured as follows:  Sec. \ref{fundamentals} gives a brief description of the relativistic dynamics of quantum systems for 
massless scalar fields.   A description of a CPTP map and its properties is given in Sec. \ref{dynamics}.  In Sec. \ref{relativisticdynamics} we 
derive a CPTP map for multipartite qudit system.  The impact of the Unruh effect on a generic quantum resource theory is analyzed in Sec. 
\ref{resource} using an open quantum 
approach utilizing the Stinespring dilation theorem.  We then investigate the free states, free operations and resource quantifiers. 
The effect of noninertial motion on free states is 
given in Sec. \ref{freestates}.  Section \ref{freeoperations} investigates the influence of noninertial motion on free operations.  The 
validity of resource quantifiers under noninertial motion is discussed in Sec. \ref{resourcequantifier}.
Finally we present our conclusions in 
Sec. \ref{conclusions}.

\section{Relativistic dynamics of bosonic modes}
\label{fundamentals}

Minkowski coordinates --- a combination of three space and one time dimension ---  are used to describe events for an inertial observer. To describe events in non-inertial frames we must use  Rindler coordinates. These two coordinate systems can be related via a Rindler transformation. In the present study, we consider space-time with one spatial and one time dimension ($z,t$). Minkowski space corresponding 
to this situation is divided into four wedges as shown in Fig. \ref{fig1}.
The regions $F$ and $P$ are the future and past light cones respectively, and the equations $|z|=t$ and $|z|=-t$ describe the future and past event horizons. The two causally disconnected Rindler regions are denoted by I and II in Fig. \ref{fig1}. In this $(z,t)$ space, the world line of a uniformly accelerated observer corresponds to a hyperbola. Each branch of the hyperbola corresponds to one Rindler region. The coordinates of the two Rindler regions are 
\begin{align}
  t &=  a^{-1} e^{a \xi} \sinh a \tau, \nonumber \\
  z& =a^{-1} e^{a \xi} \cosh a \tau 
\end{align}
for $|z|<t$ and 
\begin{align}
  t &= -a^{-1} e^{a \xi} \sinh a \tau, \nonumber \\
  z& =a^{-1} e^{a \xi} \cosh a \tau 
  \end{align}
for $ |z|>t $.  Here, $a$ is the acceleration, $\tau$ is the proper time, and $\xi$ is the space-like coordinate.

\begin{figure}
\includegraphics[width=0.69\columnwidth]{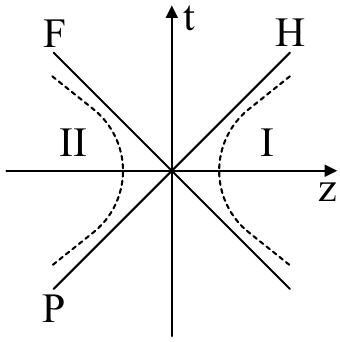} 
  \caption{Minkowski space represented by the $(z,t)$ plane is divided into four regions of Rindler coordinates. In this figure, $H$ denotes the horizon. The future $F$ and past $P$ event horizons are represented by solid lines and the accelerated observers path are represented by the dashed lines.}
\label{fig1}
\end{figure}

Consider a free massless scalar field in Minkowski space-time. Let Alice and Bob each have a monochromatic detector with frequencies 
$\omega_1$ and $\omega_2$, respectively. In this framework, a Bell state reads
\begin{equation}
\frac{1}{\sqrt{2}}\bigg({|0_{\omega_{1}} \rangle^{\mathcal{M}} |0_{\omega_{2}} \rangle^{\mathcal{M}} + |1_{\omega_{1}} \rangle^{\mathcal{M}} |1_{\omega_{2}} \rangle^{\mathcal{M}}}\bigg),
\end{equation}
where $|0_{\omega_i}\rangle^\mathcal{M}$ and $|1_{\omega_i}\rangle^\mathcal{M}$ are the vacuum state and first excited state with frequency $\omega_i$ in Minkowski space. In an inertial frame, the vacuum states are the respective ground states. The excited states can be obtained from these ground states by applying the corresponding Minkowski creation operators as
\begin{align}
|n_{\omega_{i}} \rangle^{\mathcal{M}} = \frac{(\hat{a}_{\omega_{i}}^{\dag})^{n}}{\sqrt{n !}} |0_{\omega_{i}} \rangle^{\mathcal{M}}.
\end{align}
Let us consider a situation where one of the particles, say Bob's, moves with a uniform acceleration. The quantum state of Bob can be specified using the Rindler or Unruh coordinates since the wavefunction is highly delocalized in space. The initial Minkowski space is divided into two causally disconnected regions in the Rindler and Unruh basis. The Unruh modes are sharply peaked at the Rindler frequency $\Omega=\frac{|
\omega|c}{a}$, which is a dimensionless quantity. Here $\omega$ and $c$ are the wave vector and velocity of light and $a$ is the acceleration of Bob. The Minkowski and the Rindler modes can be related through the Unruh modes as follows:
\begin{align}
\phi &= \int_{0}^{\infty} (a_{{\omega},\mathcal{M}} 
          u_{{\omega},\mathcal{M}} + a^{\dagger}_{{\omega},\mathcal{M}} u^*_{{\omega},\mathcal{M}}) \; d\omega \nonumber \\
     &= \int_{0}^{\infty} \bigg( A_{{\Omega},R} 
          u_{{\Omega},{R}} + A^{\dagger}_{{\Omega},R} u^*_{{\Omega},R} + A_{{\Omega},L} u_{{\Omega},{L}}  \nonumber  \\
     &  \phantom{aaaaaa}  + A^{\dagger}_{{\Omega},L} u^*_{{\Omega},L}             \bigg) \; d\Omega \nonumber \\
     &= \int_{0}^{\infty} \bigg( b_{{\Omega},\text{I}} 
          u_{{\Omega},{\text{I}}} + b^{\dagger}_{{\Omega},\text{I}} u^*_{{\Omega},\text{I}} + b_{{\Omega},\text{II}} u_{{\Omega},\text{II}} \nonumber  \\
     &   \phantom{aaaaaa}  + b^{\dagger}_{{\Omega},\text{II}} u^*_{{\Omega},\text{II}} \bigg) \; d\Omega .
\end{align}
Here, $a_{\omega,\mathcal{M}}$ is the Minkowski annihilation operator, $A_{\Omega,R}  $ and $ A_{\Omega,L}$ are the Unruh annihilation operators for the right ($R$)  and left ($L$) regions, and $b_{\Omega,\text{I}}$ and $b_{\Omega, \text{II}}$ are the Rindler annihilation operators in regions I and II. These operators obey bosonic commutation relations:
\begin{align}
    [a_{\omega_1,\mathcal{M}}, a^\dagger_{\omega_2,\mathcal{M}}]
       &= \delta_{\omega_1\omega_2} , \nonumber \\
    [A_{\Omega_1,R},A^\dagger_{\Omega_2,R}]
       &= [A_{\Omega_1,L},A^\dagger_{\Omega_2,L}]
       =\delta_{\Omega_1\Omega_2} , \nonumber \\
   [b_{\Omega_1,\text{I}},b^\dagger_{\Omega_2,\text{I}}]
       &= [b_{\Omega_1,\text{II}},b^\dagger_{\Omega_2,\text{II}}]
       =\delta_{\Omega_1\Omega_2}  .
\end{align}
The commutation relations between operators in the $R$ and $L$ regions corresponding to the Unruh modes and the operators in region I and II of Rindler 
coordinates vanish. Since the creation and annihilation operators of the Minkowski and Unruh bases do not mix and we have $\ket{0}_\mu = \ket{0}_U = \prod_\Omega\ket{0_\Omega}_U$. 
Contrarily, for the Rindler modes, the creation and annihilation operators mix and so we have
\begin{align}
 \ket{0_{\Omega}}_U = \sum_N \frac{\tanh^n r_{\Omega}}{\cosh r_{\Omega}} \ket{n_{\Omega}}_{\text{I}} \ket{n_{\Omega}}_{\text{II}},
 \end{align}
where $\ket{n_\Omega}_{\text{I}}$ is the $n$th excited state of the Rindler I vacuum state.
For our study, we consider a wave packet narrowly peaked in $\Omega$. Consequently the Unruh and Rindler commutators are 
$[A_{\Omega,R},A^\dagger_{\Omega,R}]
       = [A_{\Omega,L},A^\dagger_{\Omega,L}]=1$ and 
$[b_{\Omega,{\text{I}}},b^\dagger_{\Omega,{\text{I}}}]
       = [b_{\Omega,{\text{II}}},b^\dagger_{\Omega,{\text{II}}}]=1$.
Given these idealized circumstances, the most general creation operator corresponding to a purely positive Minkowski frequency is 
\begin{align}
a^\dagger_{\Omega, U} = q_LA^{\dagger}_{\Omega,L} + q_RA^{\dagger}_{\Omega,R},
\end{align}
where $q_R$ and $q_L$ are complex numbers with $|q_R|^2 + |q_L|^2 = 1$. Considering $q_R=1$ and $q_L=0$, we have
\begin{align}
 a^{\dagger}_{\Omega,U}\ket{0_{\Omega}}_U
     = \sum_{n=0}^{\infty} \frac{\tanh^n r_{\Omega}}{\cosh r_\Omega} \bigg(\frac{\sqrt{n+1}}{\cosh r_{\Omega}} \bigg) \ket{\Phi^n_\Omega}~, 
\end{align}   
where 
\begin{align}
 \ket{\Phi^n_\Omega} = q_L\ket{n_\Omega}_{\text{I}}\ket{(n+1)_\Omega}_{\text{II}}
+ q_R\ket{(n+1)_\Omega}_{\text{I}}\ket{n_\Omega}_{\text{II}}~.
\end{align}
Using the monochromatic wave approximation, the Minkowski and Rindler modes are related as 
\begin{align}
 \hat{a}^{\dagger}_{\omega} &= \hat{b}^{\dagger}_{\Omega_{\text{I}}}
                                 \cosh r- \hat{b}_{\Omega_{\text{II}}}\sinh r =\hat{S}_{\Omega} \hat{b}^{\dagger}_{\Omega_{I}}
                                \hat{S}^\dagger_{\Omega}~,\nonumber \\
 \hat{a}_{\omega}          &= \hat{b}_{\Omega_{\text{I}}}\cosh r -
                                \hat{b}^\dagger_{\Omega_{\text{II}}}\sinh r 
                              = \hat{S}_{\Omega} \hat{b}_{\Omega_{\text{I}}}
                                \hat{S}^\dagger_{\Omega}~,
\end{align}
where 
\begin{align}
    \hat{S}_\Omega(r)=\exp[r(b^\dagger_{\Omega \text{I} }b^\dagger_{\text{II}} -b_{\Omega I}b_{\text{II}})] ,
    \label{Stransform}
\end{align}
performs the transformation from Minkowski coordinates 
to Rindler coordinates. Consequently, the single mode Minkowski vacuum corresponding to a non inertial observer becomes
\begin{align}
 \ket{0_{\omega}}^{\mathcal M} &= \hat{S}_{\Omega}(r)               (\ket{0}_{\text{I}} \otimes \ket{0}_{\text{II}}),              \nonumber \\
                              &= \frac{1}{\cosh r} \sum_{n=0}^\infty \tanh^n{r}\ket{n_\Omega}_{\text{I}} \ket{n_\Omega}_{\text{II}} ,
\end{align}
where $\cosh{r} = (1-e^{-2\pi\Omega})^{-1/2}$ with $\ket{n_\Omega}_{\text{I}} $ and $\ket{n_\Omega}_{\text{II}} $ being the mode decomposition in Rindler regions I and II respectively. The first excited state in Minkowski coordinates is 
\begin{align}
 \ket{1_\omega}^{\mathcal{M}} &= \hat{a}^\dagger_{\omega} \ket{0_{\omega}}^{\mathcal M},  \nonumber \\
                              &= \hat{S}_{\Omega}(r)  \hat{b}^\dagger_{\Omega \text{I}} (\ket{0}_{\text{I}} \otimes \ket{0}_{\text{II}}),              \nonumber \\
                              &= \frac{1}{\cosh^2 r} \sum_{n=0}^\infty \sqrt{n+1}\tanh^n{r}\ket{(n+1)_\Omega}_{\text{I}} \ket{n_\Omega}_{\text{II}}.
\end{align}
We observe that, due to the squeezing behavior of a non inertial observer, a single Minkowski mode can be written as a superposition of two Rindler modes. Consequently, there exists coherence between the two Rindler modes. Since these two modes are not causally connected, this coherence cannot be experimentally observed.

\section{Noisy quantum channels}
\label{dynamics}

In this section we briefly summarize the formalism associated with noisy quantum channels, and the properties that must be followed to be called a CPTP map. 
    
Consider a quantum system $Q$ with an initial state $\rho$ in contact with an external environment.  In the formalism that we develop in this paper, all operators act with respect to the system $ Q $, hence we will only explicitly add the label $ Q $ when there is risk of confusion.  A dynamical evolution of the state $\rho$ is described by the map $\mathcal{E}$ such that the evolution is 
\begin{align}
   \rho \rightarrow \rho' = \mathcal{E}(\rho)~.
\end{align}
The map $\mathcal{E}$ characterizing the evolution should satisfy the following properties in order that it is a CPTP map.  
\begin{enumerate}
  \item[(i)] {\it  Linearity:} The map $\mathcal{E}$ must act linearly on the density matrices, meaning that if \\
                        $\rho = p_1\rho_1 + p_2 \rho_2$ then
\begin{align}
   \mathcal{E}(\rho) 
                         &= p_1 \mathcal{E}(\rho_1) + p_2 
                         \mathcal{E} (\rho_2).
\end{align}
This property implies that the channel preserves a probabilistic combination of input states, ensuring a correct evolution of mixed states.
\item[(ii)] {\it Trace Preserving:} 
  For a map $\mathcal{E}$ to be trace preserving, it must satisfy
  ${\rm Tr} (\rho) = {\rm Tr}(\mathcal{E} (\rho) ) $.
  The trace-preserving nature ensures that the quantum state stays normalized under evolution.
\item[(iii)] {\it Positive:}
  Since the density matrix $\rho$ is positive, the evolved state $\rho'=\mathcal{E}(\rho)$ must also be positive. Here, we use the word positive to refer to positive semidefinite matrices which are Hermitian and do not have negative eigenvalues. The positivity criterion assures that the output is also a valid quantum state.
  Along with satisfying the conditions (i), (ii) and (iii)
  described above, the superoperator $\mathcal{E}$ should also satisfy the criterion of complete positivity described below. 
\item[(iv)] {\it Completely positive:} 
The superoperator $\mathcal{E}(\rho)$ should be completely positive.
Complete positivity entails the following properties. Let us assume that the superoperator is extended to apply over a compound system $PQ$ where $\mathcal{E}^{PQ} = I^P \otimes \mathcal{E}^Q$, where $I^P$ is the identity superoperator acting on $P$ and $\mathcal{E}^Q $ is the superoperator acting on $Q$. The superoperator $\mathcal{E}^Q(\rho)$ is completely positive if the superoperator of the form $I^P \otimes \mathcal{E}^Q$ is also positive. A completely positive map ensures that when the size of the quantum system increases the superoperator is extensible. This property implies that quantum operations remain valid even when applied to a subsystem in an entangled state. 
If the map $\mathcal{E}^Q$ is given by 
\begin{align}
\mathcal{E}^Q(\rho)=\sum_k A_k^Q  \rho   {A_k^Q}^{\dagger}
\label{cptpmap}
\end{align}
and $A_k^Q$ is an operator on the Hilbert space of $ Q $ then the map is a completely positive map.
\end{enumerate}

\section{Noninertial dynamics as a CPTP map}
\label{relativisticdynamics}

\subsection{Kraus operator}

We propose a CPTP map to describe the Unruh effect in a $N$-partite system of which $M$ parties 
exhibit non-inertial motion.  This allows for a compact representation of how a quantum system is altered due to the presence of subsystems in non-inertial 
systems.  

\begin{theorem} \label{theorem1}
For a $N$-partite system in which $M$ subsystems are under non-inertial motion, the Unruh effect can be represented by 
a CPTP map given by
\begin{align}
{\cal E} (\rho') = \sum_k A_k  \rho'  A_k^{\dagger}
\label{maptheorem1}
\end{align}
and $ \rho'$ is the state where all $ N $ parties are in inertial motion. Here, the Kraus operator $A_{k}$ is given by 
\begin{equation}
    A_{k} = {\mathcal{I}} \otimes {\mathcal{A}}_{k}
    \label{Maprelativistic}
\end{equation}
where $\mathcal{I}$ is the identity operator which acts on the sector of $N-M$ inertial qudits and ${\mathcal{A}}_{k}$ is the operator acting on the 
noninertial qudits.  Here, $\mathcal{I}$ and ${\mathcal{A}}_{k}$ are given by
\begin{equation}
    {\mathcal{I}} = \bigotimes_{m=1}^{N-M} I_{m}; \qquad \qquad
    {\mathcal{A}}_{k} = \bigotimes_{m=1}^{M} {\mathsf{A}}_{k_{m}}.
    \label{quantitiesIandA}
\end{equation}
The Kraus operator label $ k $ labels each combination of $ (k_1,k_2, \dots, k_M) $, Here ${\mathsf{A}}_{k_{m}}$ defines the evolution of one qudit under non-inertial motion and is given by 
\begin{equation}
  {\mathsf{A}}_{k_m} =  \frac{(b^{\dagger}_{m})^{k_m}} { \sqrt{k_m !}} 
                      \frac{\tanh^{k_m} r_{m}}{\cosh^{\Hat{\ell}_{m} +1 } r_{m}}  ,
\label{RelativisticmapA}                      
\end{equation}
where $\Hat{\ell}_m = b^\dagger_m b_m$ is the number operator acting on the accelerating system's Hilbert space. 
\end{theorem}

\begin{proof} 
When $M$ of the $N$  parties are under non-inertial motion, the quantum states 
corresponding to these parties are described by Rindler coordinates in Minkowski space-time.  For any given quantum party, the $\ell $th excited state for a $d$-dimensional qudit $ \ell \in [0, d-1] $ is
\begin{eqnarray}
|\ell_{\omega} \rangle^{\cal M} &=& \hat{S}_{\Omega} \frac{( b_{\text{I}}^{\dagger} )^{\ell} }{\sqrt{\ell !}} ( |0 \rangle_{\text{I}} \otimes |0 \rangle_{\text{II}} ) \nonumber \\
                          &=& \frac{1}{\sqrt{\ell!}} \; \frac{1}{\cosh^{\ell+1} r} \sum_{n=0}^{\infty} \tanh^{n} r \sqrt{(n+1) \cdots (n+\ell)} \nonumber \\
                          & & |n+\ell \rangle_{\text{I}} \otimes |n \rangle_{\text{II}} 
    \label{excitedstate} 
\end{eqnarray}
where $\hat{S}_{\Omega}$ is the coordinate transformation described in Eq. (\ref{Stransform}).  Let us consider a general $N$-partite quantum state
\begin{equation}
    |\psi \rangle = \sum_{\ell_1 \dots \ell_N=0}^{d-1} \alpha_{\ell_1 \dots \ell_N} 
                   | \ell_1 \dots \ell_N \rangle^{\cal M}
\label{quantumstate}                   
\end{equation}
where $\alpha_{\ell_1 \dots \ell_N} $ is the coefficient of superposition and $\sum_{\ell_1 \dots \ell_N} |\alpha_{\ell_1 \dots \ell_N}|^{2} = 1$. Here we have changed notation where the frequency label $ |\ell_{\omega} \rangle  $ has been converted to a mode label $ |\ell_{m} \rangle $.  Now when $M$ parties undergo acceleration, the quantum state becomes 
\begin{align}
&  |\psi \rangle = \sum_{\ell_1 \dots \ell_N=0}^{d-1}  \alpha_{\ell_1 \dots \ell_N} \Bigg[  \bigotimes_{m'=1}^{N-M} | \ell_{m'} \rangle \Bigg] \nonumber \\
  &   \bigotimes_{m=1}^{M} \sum_{n_{m} = 0}^{\infty}
  \frac{ \tanh^{n_{m}} r_{m}}{\cosh^{\ell_{m} + 1} r_{m}}
\sqrt{\frac{(n_{m}+\ell_{m})!}{\ell_{m}! n_{m} !}}
                    |n_{m} +\ell_{m} \rangle_{\text{I}} \otimes | n_{m}  \rangle_{\text{II}}  ,
\label{acceleratedstate}                    
\end{align}
using (\ref{excitedstate}). From the expression of the accelerated state in Eq. (\ref{acceleratedstate}) we can construct the corresponding density matrix.  
The two Rindler modes in the density matrix are causally disconnected, hence we can trace out the second Rindler mode. 
The density matrix after tracing out the second Rindler mode is
\begin{align}
 & \rho = \sum_{\ell_1 \dots \ell_N=0}^{d-1} \sum_{\ell_1' \dots \ell_N' =0}^{d-1} \alpha_{\ell_1 \dots \ell_N}  \alpha_{\ell_1' \dots \ell_N'}^* \sum_{n_{m} = 0}^{\infty} \Bigg[  \bigotimes_{m'=1}^{N-M} |\ell_{m'}  \rangle  \langle \ell_{m'}' |  \Bigg]  \nonumber \\
& \Bigg[    \bigotimes_{m=1}^{M} 
 \frac{\tanh^{2n_{m}} r_{s}}{  \cosh^{\ell_{s} + \ell_{m}' + 2} r_{s}}
\sqrt{\frac{(n_{m}+\ell_{s})! (n_{m} +\ell_{m}')!  }{ \ell_{s}! \ell_{m}'!  (n_{m} !)^2 }}         \nonumber \\
& |n_{m} +\ell_{s} \rangle_{\text{I}}  \langle n_{m} +\ell_{m}' |_{\text{I}} \Bigg]  .
\label{densitymatrixaccelerated}       
\end{align}
Using the definition (\ref{RelativisticmapA}) we evalaute
\begin{align}
{\mathsf{A}}_{k_m}  | \ell_m \rangle =   \frac{ \tanh^{k_m} r_{m}}{\cosh^{\ell_m +1} r_{m}}  \sqrt{\frac{(\ell_m+ k_m)!}{k_m ! \ell_m!}} 
| \ell_m \rangle .
\end{align}
Then evaluating
\begin{align}
{\cal E} (| \psi \rangle \langle \psi | )  & =  \sum_k A_k | \psi \rangle \langle \psi | A_k^\dagger \nonumber \\
& = 
\sum_{k_1 \dots k_M=0}^\infty  {\cal I} \otimes (\bigotimes_{m=1}^M {\mathsf{A}}_{k_m} )  | \psi \rangle \langle \psi |  {\cal I} \otimes (\bigotimes_{m=1}^M {\mathsf{A}}_{k_m}^\dagger ) 
\end{align}
we obtain the same expression as (\ref{densitymatrixaccelerated}) with the replacement of dummy indices $ n_m \rightarrow k_m $.

For mixed initial state, the same logic follows as an arbitrary density matrix can be diagonalized to a probabilistic mixture of pure states.  The result for (\ref{densitymatrixaccelerated}) holds for each of the pure components.  Likewise, due to the linearity of the CPTP map, each of the pure state components mix according to their probabilities.   This proves the claim of Theorem \ref{theorem1}.  
\end{proof}

We note that the Kraus operators (\ref{RelativisticmapA}) have the same form as the bosonic amplifier channel 
\cite{ivan2011operator,qi2017capacities}.  For comparison, we write the operator in terms of its Fock state expansion
\begin{align}
{\mathsf{A}}_{k} = \frac{1}{\cosh r} \sum_{n=0}^\infty \sqrt{\binom{n+k}{k}}
\frac{\tanh^k r }{\cosh^n r} |n+k \rangle \langle n | .
\end{align}
The Kraus operator $ {\mathsf{A}}_{k}  $ corresponds to adding $ k $ bosons to the system.  

We note that the Kraus operators (\ref{RelativisticmapA}) and the CPTP map that we have derived do not agree with the results of Ref. 
\cite{ahn2018unruh} in general.  
Our results coincide for particular states such as $( | 01\rangle + |10 \rangle )/\sqrt{2} $. However as shown in Appendix
\ref{app:ahn}, for other states such as $( | 00\rangle + |11 \rangle )/\sqrt{2} $, the Kraus operator in Ref. \cite{ahn2018unruh}  is no 
longer a trace preserving map.  In contrast, our operator is trace preserving for general states as shown through the proofs in the 
next subsection.

\subsection{Properties of the Map}

Next we verify the four fundamental properties of a CPTP map as given in Sec. \ref{dynamics} for our map (\ref{maptheorem1}).

\subsubsection{Linearity}
\label{sec:prop1}

Let $\rho$ be a density matrix that is a probabilistic mixture of two density 
matrices $\rho_{1}$ and $\rho_{2}$, i.e. $\rho = p_{1} \rho_{1} + p_{2} \rho_{2}$, where $p_{1}$ and the $p_{2}$ are the mixing probabilities. 
The action of the map $\mathcal{E}$ on the density matrix is 
\begin{eqnarray}
\mathcal{E}(\rho) &=& \sum_{k} A_{k} \rho A^{\dagger}_{k} = \sum_{k} A_{k} (p_{1} \rho_{1} + p_{2} \rho_{2}) A^{\dagger}_{k}, \nonumber \\
                      &=&  \sum_{k} ( p_{1} A_{k} \rho_{1} A^{\dagger}_{k} + p_{2} A_{k} \rho_{2}  A^{\dagger}_{k} ), \nonumber \\
                      &=&  p_{1} \mathcal{E} (\rho_{1})  + p_{2} \mathcal{E} (\rho_{2}).
\end{eqnarray}
This proves that the map $\mathcal{E}(\rho)$ is linear.

\subsubsection{Trace preserving}
\label{sec:prop2}

To prove that the map ${\mathcal{E}}$ is trace preserving, first consider the eigenstate decomposition of an arbitrary density matrix $ \rho = \sum_j p_j | \psi_j \rangle \langle \psi_j | $, where $ \sum_j p_j = 1 $, and  $ | \psi_j \rangle $ are normalized states. The state before the map satisfies $ \text{Tr} (\rho) = 1 $.  After applying the map we have
\begin{align}
\text{Tr} ( {\mathcal{E}} ( \rho )) & =  \sum_j p_j \sum_{k} \text{Tr} (  A_{k} | \psi_j \rangle \langle \psi_j | A^{\dagger}_{k}  ) \nonumber \\
& =  \sum_j p_j \sum_{k} \langle \psi_j | A^{\dagger}_{k} A_{k} | \psi_j \rangle 
\label{traceeps}
\end{align}
For each term within the $ j $ sum, we have
\begin{align}
  &  \sum_{k} \langle \psi_j | A^{\dagger}_{k} A_{k} | \psi_j \rangle 
=  \sum_{k_1 \dots k_M=0}^\infty \langle \psi_j | ( {\cal I} \otimes  \bigotimes_{m=1}^M   {\mathsf{A}}_{k_m}^\dagger  {\mathsf{A}}_{k_m}) | \psi_j \rangle \label{tracepreservingzeroth} \\
 &= \sum_{\ell_1 \dots \ell_N=0}^{d-1} |\alpha_{\ell_1 \dots \ell_N}^{(j)} |^2 
 \prod_{m=1}^{M} \sum_{k_m = 0}^{\infty}  \frac{\tanh^{2 k_m} r_{m}}{{\cosh^{2(\ell_{m}+1)} r_{m}}}  \frac{ (k_m + \ell_{m})! }{ \ell_{m}! k_m!}                \label{tracepreservingfirst} \\  
 &= \sum_{\ell_1 \dots \ell_N=0}^{d-1} |\alpha_{\ell_1 \dots \ell_N}^{(j)} |^2  \prod_{m=1}^{M} 
    \frac{\cosh^{2(\ell_{m} + 1)} r_{m} }{\cosh^{2(\ell_{m} + 1)} r_{m}}, \nonumber \\
    &= \sum_{\ell_1 \dots \ell_N=0}^{d-1} |\alpha_{\ell_1 \dots \ell_N}^{(j)} |^2  = 1 ,
\end{align}
where in Eq. (\ref{tracepreservingzeroth}) we used the fact that the operator $ {\mathsf{A}}_{k_m}^\dagger  {\mathsf{A}}_{k_m} $ is diagonal in Fock space and the binomial theorem in Eq.  (\ref{tracepreservingfirst}).  Substituting the result into (\ref{traceeps}), we see that $ \text{Tr} ({\cal E} (\rho)) = 1 $, which shows that $ {\cal E} $ is trace preserving.

\subsubsection{Positive map}
\label{sec:prop3}

A given operator $O$ is positive if $\langle \chi | O | \chi \rangle $ is positive for an arbitrary state  $ | \chi \rangle $.  In our case we consider 
\begin{align}
\langle \chi | {\mathcal{E}} (\rho) | \chi \rangle & = 
\sum_{k} \langle \chi | A^{\dagger}_{k} \rho A_{k} | \chi \rangle \nonumber \\
& = \sum_{k} \langle \chi_k | \rho | \chi_k \rangle \ge 0 ,
\end{align}
where $ | \chi_k \rangle =   A_{k} | \chi \rangle $.  This is a non-negative quantity because $ \rho $ is a positive matrix
and hence $  \langle \chi_k | \rho | \chi_k \rangle \ge 0 $ for an arbitrary $ | \chi_k \rangle $.

\subsubsection{Completely Positive}
\label{sec:prop4}

Consider the evolution superoperator corresponding to the non-inertial dynamics for system $ Q $ as given in (\ref{cptpmap}). If we extend this superoperator to a composite system $PQ$ such that 
$\mathcal{E}^{PQ} = \mathcal{I}^{P} \otimes \mathcal{E}^{Q}$ we can write 
\begin{equation}
    \mathcal{E}^{PQ}(\rho) = \sum_{k} A^{PQ}_{k} \rho^{PQ} A^{PQ \dagger}_{k}.
\end{equation}
In the relativistic quantum channel identity operations imply the inertial motion of the qudits.  If the number of particles
in the inertial frame is increased by $T$ we have
\begin{eqnarray}
    A_{k}^{PQ} &=& \bigotimes_{m =1}^{T} {\mathcal{I}}_{m} \otimes  \bigotimes_{m' =1}^{N-M} {\mathcal{I}}_{m'} \otimes  {\mathcal{A}}_{k}^{Q}, \nonumber \\
               &=& \bigotimes_{m=1}^{N-M+T}  {\mathcal{I}}_{m} \otimes   {\mathcal{A}}_{k}^{Q}.
\end{eqnarray}
Since $A_{k}^{PQ}$ has the same form as $A_{k}^{Q}$ with only a trivial extension of the inertial qudits, the operator $\mathcal{E}^{Q}$ is completely positive.

\subsubsection{Summary}
The results of  Secs. \ref{sec:prop1}-\ref{sec:prop4} show that the the four properties to be satisfied by a CPTP map hold for Eq. (\ref{maptheorem1}).  We therefore find that the noninertial motion of a  $N$-partite multilevel quantum state can be expressed as a CPTP map as described in Theorem \ref{theorem1}.

\section{Effect on quantum resource theories}
\label{resource}
Now that we have established that non-inertial motion corresponds to a CPTP map, we may examine its consequences. In this section we will explore its effects on quantum resource theories.  

\subsection{Quantum resource theories}

Quantum systems exhibit a hierarchy of quantum correlations \cite{radhakrishnan2019basis,ma2019operational} ranging from nonlocality which 
is observed in spatially separated systems to quantum coherence which is found even in a single qubit.  The hierarchy consists of 
nonlocality, steering, entanglement, quantum discord and coherence. The heirarchy between the different correlations is
\begin{eqnarray}
 {\rm nonlocal} \subseteq {\rm steering}  \subseteq {\rm entanglement}   \subseteq  {\rm discord} \subseteq {\rm coherence}. \nonumber
\end{eqnarray}
Quantum resource theories provide a unified mathematical framework to describe the different correlations.  If $\mathcal{D}(\mathcal{H})$ is 
the set of all density matrices (i.e., all valid quantum states) in the Hilbert space $\mathcal{H}$
\begin{equation}
    \mathcal{D}(\mathcal{H}) = \{ \rho | \rho \geq 0, \; \rho = \rho^{\dagger}, \; {\rm Tr} \rho =1 \},
    \label{ddef}
\end{equation}
we may construct a resource theory to describe the different types of quantum correlations.  A quantum resource theory has three 
core ingredients namely: (a) free states, (b) free operations and (c) resource quantifiers.  In this section we examine the effect of 
noninertial dynamics on these three components of quantum resource theory.

\subsection{Stinespring dilation}

To analyze the effects of noninertial motion on a quantum systems, we employ the mathematical structure provided by the Stinespring 
dilation theorem.  Stinespring's theorem states that any completely positive trace preserving map (CPTP) representing a physically 
valid quantum process can be realized as a unitary interaction on an enlarged Hilbert space followed by a partial trace.  
In this sense, every apparently noisy or irreversible quantum evolution can be understood as arising from a reversible process 
acting on an extended system that includes additional, generally inaccessible, degrees of freedom.  In the relativistic setting, a 
quantum system confined to a region of space time may be treated as an open system with an Hilbert space $\mathcal{H}$.  When the 
system is viewed from a noninertial frame, the effective evolution perceived by the observer becomes nonunitary.  This is because 
acceleration introduces interactions with additional field modes and leads to thermal-like effects such as those associated with 
the Unruh effect.  Consequently, the accelerated observer describes state evolution using a CPTP map as shown in Eq. (\ref{maptheorem1}).
Using Stinespring's theorem, this CPTP map $\mathcal{E}$ can be written as a unitary operation on a larger Hilbert space that 
captures the inertial to noninertial transformation and an associated increase in dimensionality.  In the Rindler coordinates, this enlarged Hilbert space factorizes into two casually disconnected regions namely the Rindler mode I and Rindler mode II.  
Since an observer in uniform acceleration has access only to one of the modes, the physical state as shown by the noninertial 
observer is obtained by tracing out the other mode.  The partial trace leads to a loss of quantum correlations and coherence. 
Thus using Stinespring dilation establishes an open quantum system framework for noninertial evolution.  But unlike open 
quantum system, the Hilbert space associated with a quantum state increases during the noninertial evolution.  

By framing the noninertial evolution as a Stinespring dilation, we can establish a theoretical connection between the quantum resources 
in the inertial frame and the change in the quantum resources as seen by a noninertial observer.  In the following discussion, we 
explain how the fundamental components of a resource theory are modified under noninertial evolution.

\begin{figure}
\includegraphics[width=\linewidth]{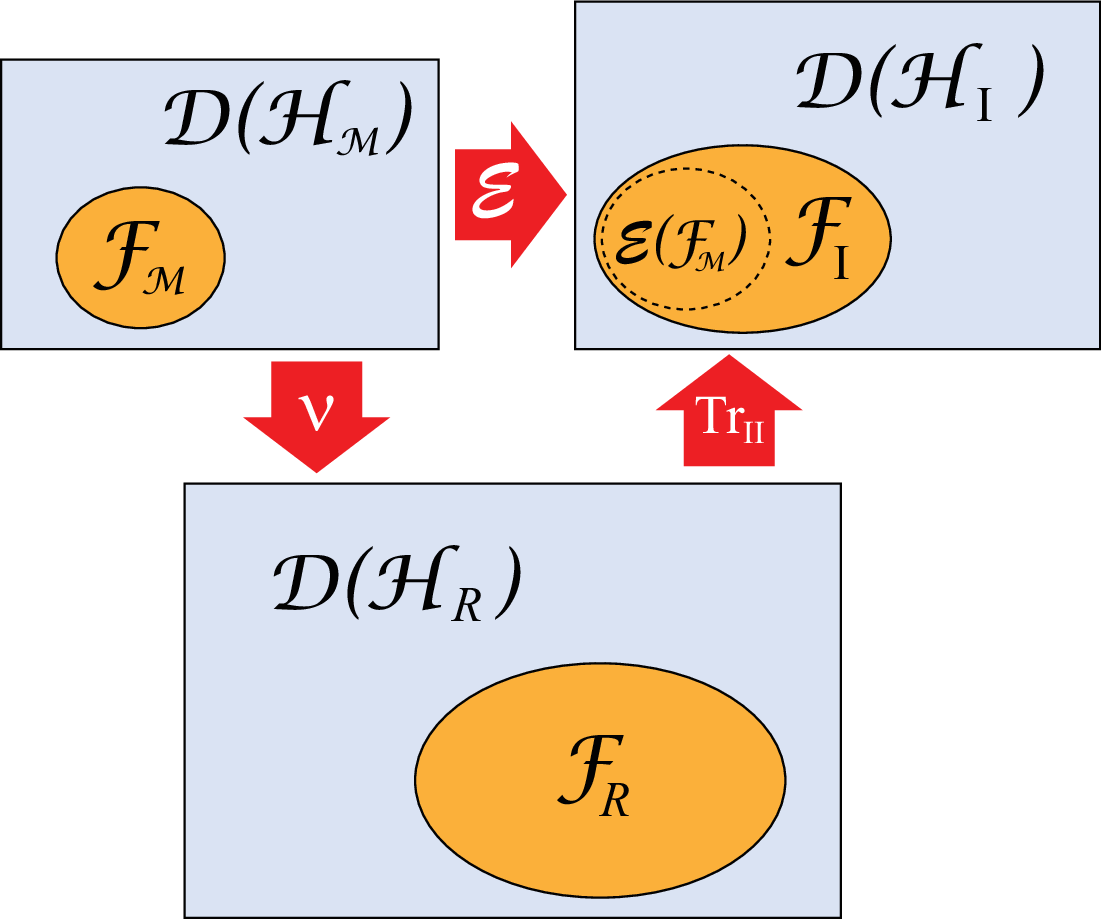} 
  \caption{Relationship between the operator spaces and free states under inertial and noninertial motion.  The top left corresponds to the operator space $ \mathcal{D}(\mathcal{H}_{\mathcal{M}}) $ for the system under inertial motion.  Within it are the free states $ \mathcal{F}_{\mathcal{M}}$.  Under Stinespring dilation, 
  the operator space expands to $\mathcal{D}( \mathcal{H}_{R}) $.  Tracing out Rindler mode II results in the operator space for Rindler mode I $ \mathcal{D}(\mathcal{H}_{\mathrm I}) $ with its associated free state $ \mathcal{F}_{\mathrm I} $.  We show the case where $ \mathcal{E} $ is a free operation of the resource theory (NRNG map), such that the mapped free states $ \mathcal{E}( \mathcal{F}_{\mathcal{M}})  $ are within  $ \mathcal{F}_{\mathrm I} $. }
\label{fig2}
\end{figure}

\section{Free states under noninertial motion}
\label{freestates}

In the framework of quantum resource theory, free states are those quantum states that do not possess the particular resource under 
consideration.  A generic free state associated with a Hilbert space $\mathcal{H}$ is denoted by $ \mathcal{F}  \subseteq \mathcal{D}(\mathcal{H}) $,  where $ \mathcal{D} $ is defined in (\ref{ddef}) and is the operator space of density matrices for a given Hilbert space. To start, let us consider an inertial system with Hilbert space $\mathcal{H}_{\mathcal{M}} $, which has an associated set of free states which have zero resource (see Fig. \ref{fig2})
\begin{equation}
    \mathcal{F}_{\mathcal{M}}  \subseteq \mathcal{D}(\mathcal{H}_{\mathcal{M}}). 
\end{equation}
For example, the set of separable states constitute the free states in the resource theory of entanglement. 

When a quantum state undergoes noninertial motion, its evolution is described by a CPTP map.  The corresponding 
Hilbert space enlarges from $\mathcal{H}_{\mathcal{M}}$ in the inertial frame to $\mathcal{H}_{R}$ in the noninertial frame (see Fig. \ref{fig2}).  
For example, if a two qubit system in an inertial frame has an Hilbert space 
$\mathcal{H}_{\mathcal{M}} = \mathbb{C}^{2} \otimes \mathbb{C}^{2}$, and one of the qubits 
undergoes noninertial motion, the Hilbert space becomes $  \mathcal{H}_{R}  = \mathbb{C}^{2} \otimes \mathbb{F} \otimes \mathbb{F} $, where $\mathbb{F}$ is the 
infinite dimensional Fock space for one mode. 
According to Stinespring's dilation formalism, the Kraus map $\mathcal{E}$ 
defined in Eq. (\ref{maptheorem1}) involves the redistribution of the initial quantum resource over the field modes.  These 
quantum resources are split between the two causally disconnected Rindler regions.  As the noninertial observer traces out one of the 
Rindler modes, part of the quantum resource becomes inaccessible.

Denoting the Hilbert space associated with the Rindler I and II regions as $ \mathcal{H}_{\mathrm I} $ and $ \mathcal{H}_{\mathrm{II}} $ respectively, after tracing out region II, the accessible region corresponds to $ \mathcal{H}_{\mathrm I} $ alone.  This space has its own set of free states which we denote as 
\begin{equation}
    \mathcal{F}_{\mathrm I} \subseteq \mathcal{D}(\mathcal{H}_{\mathrm I}). 
\end{equation}
We note that $ \mathcal{F}_{\mathrm I} $ is defined purely with respect to the space $ \mathcal{D}(\mathcal{H}_{\mathrm I}) $ (i.e. not with respect to the mapping $\mathcal{E}$).

In our discussion below, we will be primarily concerned with resource theories where the noninertial motion does not generate the resource.  

\begin{definition}[Noninertial Resource Nongenerating (NRNG)]
\label{theorem3}
Consider a CPTP map
\begin{equation}
   \mathcal{E}(\rho) = \sum_{k} A_{k} \rho A_{k}^{\dagger}, \qquad \sum_{k} A_{k} A_{k}^{\dagger} = I
\end{equation}
of the form given by Theorem \ref{theorem1}.  If  $\mathcal{E}(\mathcal{F}_{\mathcal{M}}) \subseteq \mathcal{F}_{\mathrm{I}} $  then the resource theory is noninertial and resource nongenerating (NRNG). 
\end{definition}

We now explore the relationship between the original inertial free states $ \mathcal{F}_{\mathcal{M}} $ in a NRNG resource theory and how they map under $\mathcal{E}$.

\begin{theorem}[Free operations under noninertial motion]
\label{theorem2}

Let $\mathcal{H}_{\mathcal{M}}$ denote the inertial Hilbert space, and let 
\[
\mathcal{V} : \mathcal{H}_{\mathcal{M}} \longmapsto \mathcal{H}_{\mathrm I} \otimes  \mathcal{H}_{\mathrm{II}},
\]
be an isometry implementing the Stinespring dilation associated with the inertial to noninertial transformation.  The physically 
observable noninertial channel is obtained by tracing out the inaccessible Rindler region {\rm II} according to 
\[
\mathcal{E}: \mathcal{D}(\mathcal{H}_{\mathcal{M}}) \longmapsto \mathcal{D}(\mathcal{H}_{\mathrm I}),  \qquad
\mathcal{E} (\rho) = {\rm Tr_{\mathrm{II}}} \left[ \mathcal{V} \rho \mathcal{V}^{\dagger}    \right] .
\]
Let $\mathcal{F}_{\mathcal{M}} \subseteq \mathcal{D}(\mathcal{H}_{\mathcal{M}})$ and $\mathcal{F}_{\mathrm{I}} \subseteq \mathcal{D}(\mathcal{H}_{\mathrm{I}} )$
denote the set of free states of a quantum resource theory in the inertial and noninertial description respectively.  Then 
\begin{eqnarray}
\mathcal{E} (\mathcal{F}_{\mathcal{M}}) \subseteq \mathcal{F}_{\mathrm{I}}  \Longleftrightarrow & &   \mathcal{E}  \; \hbox {is a free operation} \nonumber
\\
& &   \hbox {of the resource theory} . \label{accessibleresource}
\end{eqnarray}
Moreover, if $\exists \rho \in \mathcal{F}_{\mathcal{M}}$ for which
\begin{equation}
    {\rm Tr_{II}} \left[ \mathcal{V} \rho \mathcal{V}^{\dagger} \right] \notin \mathcal{F}_{\mathrm{I}} 
\end{equation}
then the failure to preserve free states arises solely from the partial trace over the inaccessible region {\rm II}.
\end{theorem}

\begin{proof}
    According to the basic axioms of quantum resource theories,  an operation is free iff it maps every input free state to an output state 
which is also free.  Hence we have a biconditional relation as shown in Eq. (\ref{accessibleresource}).  

For the second part we note that the isometry $\mathcal{V}$ is reversible on its image and therefore cannot destroy global correlations. 
As a matter of fact it only redistributes them between the two Rindler regions.  Meanwhile the partial trace over the inaccessible region
$\mathcal{H}_{\mathrm{II}}$ introduces mixedness and is irreversible.  Hence the accessible correlations are reduced only due to the partial trace 
operation.  
\end{proof}

\begin{theorem}[Geometry of free states under noninertial maps]
\label{theorem2.0}
Let $\mathcal{F}_{M} \subset \mathcal{D}(\mathcal{H}_{M})$ and 
$\mathcal{F}_{\mathrm{I}} \subset \mathcal{D}(\mathcal{H}_{\mathrm{I}})$ denote the set 
of free states of a quantum resource theory in the inertial and noninertial descriptions respectively.  
The set of free states is classified into (a) affine set (b) convex set and (c) nonconvex set depending on the resource theory as given 
below:
\begin{eqnarray}
& & \hbox {Affine set:} \nonumber \\
& & \forall  \rho \text{ s.t.} \; \rho  = \rho_{1} + (1 - \lambda) \rho_{2}; \; 
  \rho_{1}, \rho_{2} \in \mathcal{F},\ \lambda \in \mathbb{R}  \Rightarrow  \rho \in \mathcal{F} \nonumber  \quad\\
& & \hbox{Convex set:} \nonumber \\
& & \forall \rho_{1}, \rho_{2} \in \mathcal{F} \Rightarrow p \rho_{1} + (1 - p) \rho_{2} \in \mathcal{F}, \; p \in [0,1]
\nonumber \qquad \quad \\
& & \hbox{Nonconvex set:} \nonumber \\
& & \exists \rho_{1}, \rho_{2} \in \mathcal{F} \Rightarrow p \rho_{1} + (1 - p) \rho_{2} \notin \mathcal{F}, \; p \in [0,1]. 
\nonumber 
\end{eqnarray}
If $\mathcal{E}$ is a CPTP map describing the inertial to noninertial transformation, then $\mathcal{E}(\mathcal{F}_{\mathcal{M}})$ 
preserves the geometric structure of the free states iff $\mathcal{E}$ is a free operation of the resource theory.  
\end{theorem}

\begin{proof}
    Let $\mathcal{E}: \mathcal{D}(\mathcal{H}_{M}) \mapsto \mathcal{D}(\mathcal{H}_{\mathrm{I}})$ be an arbitrary CPTP map and not necessarily a 
free operation.  In this situation, there will be at least one free state $\rho \in \mathcal{F}_{M}$ such that 
$\mathcal{E}(\rho) \notin \mathcal{F}_{\mathrm{I}}$.  Hence $\mathcal{E}(\mathcal{F}_{\mathcal{M}})$ cannot be identified with the free states in 
the noninertial description.  Consequently, the set $\mathcal{E}(\mathcal{F}_{\mathcal{M}})$ is not constrained to retain the geometric
properties of $\mathcal{F}_{\mathcal{M}}$.  Hence convexity or affinity may be lost.  On the contrary if $\mathcal{E}$ is a free operation
i.e., $\mathcal{E}(\mathcal{F}_{\mathcal{M}}) \subseteq \mathcal{F}_{\mathrm{I}}$, then the image of any allowed combination of free states remain 
free.  Since $\mathcal{E}$ is a CPTP map, the convex and linear combinations of states in $\mathcal{F}_{\mathcal{M}}$ are mapped to 
corresponding combinations in $\mathcal{F}_{\mathrm{I}}$.  Therefore the geometric structure of the free states is preserved.  
\end{proof}

\section{Impact of Noninertial motion on Free operations}
\label{freeoperations}

The free operations in any quantum resource theory are CPTP maps which map a free state to another 
free state. If $\Phi$ is a free operation then
\begin{equation}
    \Phi(\mathcal{F}) \subseteq \mathcal{F},
\end{equation}
where $\mathcal{F} \subseteq \mathcal{D}(\mathcal{H})$ is the set of free states.  In this section, we discuss the effect of noninertial motion on free operations.  

Let $\mathcal{F}_{\mathcal{M}} \subseteq \mathcal{D}(\mathcal{H}_{\mathcal{M}})$ and 
$\mathcal{F}_{\mathrm{I}} \subseteq \mathcal{D}(\mathcal{H}_{\mathrm{I}} )$ represent the 
free sets in the inertial and noninertial frames respectively.  Consider a CPTP map in the inertial frame 
\begin{equation}
    \Phi_{\mathcal{M}}(\rho) = \sum_{n} K_{n} \rho K_{n}^{\dagger},   \qquad \sum_{n} K_{n}^{\dagger} K_{n} = I
\end{equation}
that satisfies $\Phi_{\mathcal{M}}(\mathcal{F}_{\mathcal{M}}) \subseteq \mathcal{F}_{\mathcal{M}}$, the axiom of a free operation.  

We establish the following theorems for free operations under relativistic motion.

\begin{theorem}[Composition with relativistic evolution]
\label{theorem4}

Let $\Phi_{\mathcal{M}}$ be a free operation on $\mathcal{H}_{\mathcal{M}}$,  $\Phi_{\mathrm{I}} $  be a free operation on $\mathcal{H}_{\mathrm{I}} $, and the  resource theory be NRNG.  Then 

\noindent(a) If $\Phi_{\mathrm{I}} (\mathcal{F}_{\mathrm{I}} ) \subseteq \mathcal{F}_{\mathrm{I}} $ then
\begin{align}
    (\Phi_{\mathrm{I}}  \circ \mathcal{E}) (\mathcal{F}_{\mathcal{M}}) \subseteq \mathcal{F}_{\mathrm{I}} 
    \label{compa}
\end{align}

\noindent (b) If $\Phi_{\mathcal{M}}(\mathcal{F}_{\mathcal{M}}) \subseteq \mathcal{F}_{\mathcal{M}}$ and 
$\mathcal{E}(\mathcal{F}_{\mathcal{M}}) \subseteq  \mathcal{F}_{R}$, then
\begin{align}
    ( \mathcal{E} \circ \Phi_{\mathcal{M}} ) ( \mathcal{F}_{\mathcal{M}} ) \subseteq \mathcal{F}_{\mathrm{I}} .
    \label{compb}
\end{align}
\end{theorem}

\begin{proof}
    {\it (a)} Let $\rho \in \mathcal{F}_{\mathcal{M}}$.  Then  by Theorem \ref{theorem3}, $\mathcal{E}(\rho) \in \mathcal{F}_{\mathrm{I}} $. Since $\Phi_{\mathrm{I}} $ preserves $\mathcal{F}_{\mathrm{I}} $, (\ref{compa}) holds.  {\it (b)} Let $\rho \in \mathcal{F}_{\mathcal{M}}$.  Since $\Phi_{\mathcal{M}}$ is free,  $\Phi (\rho) \in \mathcal{F}_{\mathcal{M}}$.  Then applying $\mathcal{E}$
and using the assumption $\mathcal{E}(\mathcal{F}_{\mathcal{M}}) \subseteq \mathcal{F}_{\mathrm{I}}$, (\ref{compb}) holds. 
\end{proof}

\begin{theorem}[Convex mixtures after embedding]
\label{theorem5}

Let $V: \mathcal{H}_{\mathcal{M}} \mapsto \mathcal{H}_{R}$ be an isometric embedding satisfying $V^{\dagger} V = I$ for which we define an embedded 
operation \cite{kraus1971general,hellwig1969pure,hellwig1970operations}
\begin{equation}
    \widetilde{\Phi}_{R} (\rho) = V (\Phi_{\mathcal{M}}(V^{\dagger} \rho V))V^{\dagger}  .
\end{equation}
Let the  resource theory be NRNG such that $\mathcal{E}$ is a free operation and $\widetilde{\Phi}_{R}$ preserve $\mathcal{F}_{R}$. Then for any $p \in [0,1]$, the convex combination 
\begin{align}
    \Lambda_{R} = p \mathcal{E} + (1-p) \widetilde{\Phi}_{R}
\end{align}
is CPTP and satisfies 
\begin{align}
    \Lambda_{R} ( \mathcal{F}_{R} ) \subseteq \mathcal{F}_{R} .
\end{align}
\end{theorem}

\begin{proof}
    The map $\widetilde{\Phi}_{R}$ is CPTP because it is a composition of CPTP maps and an isometric embedding.  Since $\Lambda_{R}$ is a 
convex combination of CPTP maps acting on the same space it is CPTP as well.  Let $\rho \in \mathcal{F}_{R}$, then 
$\mathcal{E}(\rho) \in \mathcal{F}_{R}$ and $\widetilde{\Phi}_{R} (\rho) \in \mathcal{F}_{R}$.  Since $\mathcal{F}_{R}$ is convex according
to standard axioms of quantum resource theory,
\begin{equation}
    \Lambda(\rho) = p \; \mathcal{E}_{R} (\rho) + (1-p) \widetilde{\Phi}_{R}(\rho)  \in \mathcal{F}_{R} 
\end{equation}
and hence $\Lambda(\mathcal{F}_{R}) \subseteq \mathcal{F}_{R}$.
\end{proof}

\begin{theorem}[Tensor product composition]
\label{theorem6}

Let $\Phi_{\mathcal{M}_A} : \mathcal{D}(\mathcal{H}_{\mathcal{M}_A} ) \mapsto \mathcal{D}(\mathcal{H}_{\mathcal{M}_A} )$ be a free operation 
on subsystem $A$.  Let the resource theory be NRNG such that
$\mathcal{E}: \mathcal{D}(\mathcal{H}_{\mathcal{M}_B}) \mapsto \mathcal{D}(\mathcal{H}_{\mathrm{I}_B})$ is a free operation, where $\mathcal{H}_{\mathcal{M}_B }$ and $\mathcal{H}_{\mathrm{I}_B }$
are the spaces corresponding to the inertial and noninertial nature of the system $B$.  Assume that the composite free sets satisfy 
the local to global closure condition 
\begin{align}
    \rho_{AB} \in \mathcal{F}_{\mathcal{M}_{AB} } \Rightarrow {\rm Tr}_{B} (\rho_{AB}) \in \mathcal{F}_{\mathcal{M}_A}, \; 
    {\rm Tr}_{A} (\rho_{AB}) \in \mathcal{F}_{\mathcal{M}_B},
\end{align}
then the tensor product channel 
\begin{equation}
    \Phi_{\mathcal{M}} \otimes \mathcal{E}:  
    \mathcal{D}( \mathcal{H}_{\mathcal{M}_A} \otimes \mathcal{H}_{\mathcal{M}_B} ) \mapsto  
    \mathcal{D}( \mathcal{H}_{\mathcal{M}_A} \otimes \mathcal{H}_{\mathrm{I}_B} )
\end{equation}
preserves freeness
\begin{align}
    ( \Phi_{\mathcal{M}} \otimes \mathcal{E} ) \mathcal{F}_{\mathcal{M}_{AB}} \subseteq \mathcal{F}_{\mathrm{I}_{AB}} .
\end{align}
\end{theorem}

\begin{proof}
    Let $\rho_{AB} \in \mathcal{F}_{\mathcal{M}_{AB}}$ by the closure assumption
\begin{align}
    \rho_{A} & = {\rm Tr}_{B} (\rho_{AB}) \in \mathcal{F}_{\mathcal{M}_A} \nonumber \\
    \rho_{B} & = {\rm Tr}_{A} (\rho_{AB}) \in \mathcal{F}_{\mathcal{M}_B}. 
\end{align}
Since $\Phi_{\mathcal{M}}$ preserves $\mathcal{F}_{\mathcal{M}_A}$ and $\mathcal{E}$ preserves $\mathcal{F}_{R_B}$
\begin{equation}
\Phi_{\mathcal{M}}(\rho_{A}) \in \mathcal{F}_{\mathcal{M}_A},  \qquad \mathcal{E}(\rho_{B}) \in \mathcal{F}_{\mathrm{I}_B} .
\end{equation}
The tensor product output lives in the space $\mathcal{H}_{\mathcal{M}_A} \otimes \mathcal{H}_{\mathrm{I}_B}$ and belongs to $\mathcal{F}_{\mathrm{I}_{AB}}$, by 
definition of the noninertial free set on the composite system.  
\end{proof}

\section{Validity of Resource Quantifiers under noninertial motion}
\label{resourcequantifier}

The third component of a quantum resource theory is the resource quantifier.  For a given quantum state, a quantifier assigns a measure
to the amount of resource contained in it.  Mathematically, a resource quantifier $Q(\rho)$ assigns a nonnegative real number to a
quantum state $\rho$.  Some of the fundamental properties required of a resource theory are (a) faithfulness, (b)  monotonicity,  
and (c) convexity.  Measures for resources such as entanglement, coherence, and steerability mostly satisfy these axioms.  One exception to this 
is quantum discord as defined in Ref. \cite{ollivier2001quantum,henderson2001classical,radhakrishnan2020multipartite}, which does not satisfy 
convexity.  Any given quantum resource can be measured using different quantifiers, 
where each measure corresponds to a specific mathematical formulation.  Most of these measures fall into any one of the three broad 
classifications, namely:  convex monotones, robustness-based resource quantifiers, and distance-based resource quantifiers. 
Below we discuss the effects of noninertial motion on these three different classes of resource quantifiers.

\subsection{Convex monotones}
\begin{theorem}[Convexity under noninertial transformations]
\label{theorem7} 
Consider a measure $Q(\rho)$ which is convex and monotonic under free operations $ \Phi (\rho) $.  Let the  resource theory be NRNG such that 
map $ \mathcal{E} (\rho) $ is a free operation.   
Then for a generic decomposition $\rho = \sum_{i} p_{i} \rho_{i}$  
\begin{equation}
    Q \left(\mathcal{E} \circ \Phi(\rho) \right) \leq \sum_{i} p_{i} Q (\rho_{i}).
    \label{thm7}
\end{equation}
\end{theorem}

\begin{proof}
A convex monotone is a resource quantifier which satisfies the convexity axiom
\begin{equation}
\forall \rho_{i}:  Q\left( \sum_{i} p_{i} \rho_{i} \right) \leq \sum_{i} p_{i} Q(\rho_{i}),   \qquad \forall p_{i} \in [0,1].    
\end{equation}
Convexity implies that classical randomness does not increase the amount of quantum resource in a system.  In the noninertial setting, 
convexity also ensures that the resource monotone remains consistent when we embed $\mathcal{H}_{\mathcal{M}}$ into the 
enlarged space $\mathcal{H}_{R}$. 
\begin{equation}
    Q \left(\mathcal{E} \circ \Phi (\rho) \right) = Q \left(\mathcal{E} \circ \Phi \bigg(\sum_{i} p_{i} \rho_{i} \bigg) \right). 
\end{equation}
Since $\mathcal{E}$ and $\Phi$ are CPTP operations, which are linear maps on the space of density matrices, 
\begin{equation}
  \mathcal{E} \circ \Phi \left( \sum_{i} p_{i} \rho_{i} \right)  = \sum_{i} p_{i} \mathcal{E} \circ \Phi(\rho_{i}). 
\end{equation}
On applying the convexity axiom
\begin{equation}
    Q\left(\sum_{i} p_{i} \; \mathcal{E} \circ \Phi(\rho_{i}) \right)  \leq 
    \sum_{i} p_{i} \, Q \left( \mathcal{E} \circ \Phi(\rho_{i}) \right)
\end{equation}
The operation $\Phi$ is CPTP and so $Q(\Phi(\rho)) \leq Q(\rho)$.  Similarly the noninertial map $\mathcal{E}$ is resource nongenerating
and so $Q(\mathcal{E}(\rho)) \leq Q(\rho)$.  Hence we have
\begin{equation}
    Q \left(\mathcal{E} \circ \Phi(\rho_{i}) \right) \leq Q(\rho_{i}).
\end{equation}
From this, (\ref{thm7}) follows.  
\end{proof}
Thus if a resource quantifier is a convex monotone in an inertial frame, it retains the same structure in the noninertial frame as well. 
Hence any convex monotone is a valid resource quantifier under the noninertial channel $\mathcal{E}$, provided that the channel does not 
generate any resource.

\subsection{Robustness-based resource quantifier}

Robustness-based measures are another fundamental and widely used type of quantifier in quantum resource theories.  It measures the quantum 
resource of a state by finding the amount of noise to be added to destroy the resource.  Entanglement, coherence and magic are some of 
the popular resources based on this scheme.  The generalized robustness based measure of a quantum resource for a given quantum state 
is defined as 
\begin{equation}
    R_{G}(\rho) = \min \left \{ s \geq 0: \exists \, \tau \;\; {\rm s.t.}  \;\; \frac{\rho + s \tau}{1 + s} \in \mathcal{F} \right \} ,
\end{equation}
where $\mathcal{F}$ is the set of free states. 

\begin{theorem}
 Let the resource theory be NRNG such that 
map $ \mathcal{E} (\rho) $ is a free operation. Then for any free operation $\Phi$, 
\begin{equation}
    R_{G} \left(\mathcal{E} \circ \Phi (\rho) \right) \leq R_{G} (\rho)  .  
    \label{robustness}
\end{equation}
\end{theorem}

\begin{proof}
     If a CPTP map $\Lambda$ does not generate a resource then
\begin{equation}
    \Lambda \left( \frac{\rho + s t} {1 + s}  \right) = \frac{\Lambda(\rho) + s \Lambda(t)}{1 + s} \in \mathcal{F}  .
\end{equation}
It then follows that 
\begin{equation}
    R_{G} \left( \Lambda (\rho) \right) \leq R_{G} \left( \rho \right), 
\end{equation}
from which (\ref{robustness}) follows. 
\end{proof}

This means that if the noninertial transformation $\mathcal{E}$ is a resource nongenerating map, then robustness can be used as a resource quantifier in both inertial and noninertial frames.

\subsection{Distance-based resource quantifiers}

A third class of resource quantifier are distance-based measures.  These measures are mathematical functions which estimate the 
distance between a given quantum state and the closest resource free state in the space of density matrices $\mathcal{D}(\mathcal{H})$. 

\begin{theorem}
 Let the resource theory be NRNG such that 
map $ \mathcal{E} (\rho) $ is a free operation. Then for a  free CPTP operation $\Phi$, 
\begin{equation}
   D\left( \mathcal{E} \circ \Phi(\rho), \mathcal{E} \circ \Phi(\sigma) \right) \leq D(\rho, \sigma). 
   \label{distanceunderephi}
\end{equation}
\end{theorem}

\begin{proof}
To be considered as a resource quantifier, a distance measure should be contractive.  Hence a given distance measure $D(\rho, \sigma)$ must satisfy
\begin{equation}
    D\left(\Lambda(\rho), \Lambda(\sigma) \right) \leq D(\rho, \sigma)
\end{equation}
for all free operations $\Lambda$. From this Eq. (\ref{distanceunderephi}) follows.  
\end{proof}

Thus any distance measure which can be used in an inertial frame can also be used in  noninertial frame if  
{\it(i)} The noninertial map $\mathcal{E}$ is resource nongenerating; and 
{\it(ii)} The distance $D(\rho, \sigma)$ is contractive under all CPTP maps.  

Measures such as relative entropy, trace distance, Bures distance  are contractive under CPTP maps.  Hence these distances can be used 
to estimate quantum resources both in inertial and noninertial setting.  Contrarily, the Hilbert-Schmidt distance is not contractive under 
all CPTP maps \cite{ozawa2000entanglement} and therefore cannot be used as a valid resource quantifier in general setting.

\section{Summary and Conclusions}
\label{conclusions}

We have examined the effect of noninertial motion on quantum state from the point of view of a quantum channel, and examined its consequences on quantum resource theories.  In the first part of this paper, we derived the Kraus operators corresponding to noninertial motion, and showed that it is a CPTP map.  
Our formulation is generic and can be used to describe $N$-partite multilevel systems. It is equivalent to a bosonic amplifier channel, which takes into account of the virtual thermal bath generated by the Unruh effect. We note similar investigations connecting noninertial dynamics to open quantum systems were done in Refs. \cite{ahn2018unruh,aspachs2010optimal}. However, previous formulations are not trace preserving for particular states (see Appendix \ref{app:ahn}) and hence problematic.  For the 
transformation used in Ref. \cite{aspachs2010optimal}, the necessary properties of positivity and trace preserving were not verified.    

In the second part of this paper, we use the results establishing the noninertial motion as a CPTP map to examine its effects on quantum resource theories.  This is done using Stinespring's dilation theorem, which establishes a link between the inertial quantum resources and the subsequent evolution introduced by noninertial motion.  We show 
that the loss of information to the inaccessible Rindler region is formally equivalent to the interaction between the system and an 
external bath.  We then consider the effects of the noninertial map on the three key elements namely (i) free states, 
(ii) free operations and (iii) resource quantifiers. We defined the notion of a NRNG resource theory to 
investigate its implication on free states, free operations and resource quantifiers. From the analysis we observed that the basic geometry of the free states
do not change under noninertial motion.  Then we discussed the changes brought about in free operations when it is combined with a 
noninertial CPTP map.  Finally, we checked the validity of resource quantifiers under noninertial transformations.

We anticipate that the quantum channel description of noninertial motion will be helpful to understand the effects of accelerated motion on quantum communication 
and computing, through its ease of use and familiarity in the wider quantum information community.   In this work, we have only considered flat space-time and the results derived hold only under these conditions.  A more detailed attempt to develop a quantum channel framework in curved space-time is an interesting future direction that will extend the scope of these results.

\section*{Acknowledgements}
T.B. is supported by the SMEC Scientific Research Innovation Project (2023ZKZD55); the National Natural Science Foundation of China (92576102); the Science and Technology Commission of Shanghai Municipality (22ZR1444600); the NYU Shanghai Boost Fund; the China Foreign Experts Program (G2021013002L); the NYU-ECNU Institute of Physics at NYU Shanghai; the NYU Shanghai Major-Grants Seed Fund; and Tamkeen under the NYU Abu Dhabi Research Institute grant CG008.

\bibliography{reference}

\appendix

\section{Violation of trace preserving property of the map given in Ref. \cite{ahn2018unruh}}
\label{app:ahn}

In this section we show that the map derived in Ref. \cite{ahn2018unruh} is not a universal CPTP map.  
To prove this let us consider the Bell state $| \Phi^{+} \rangle$ in the Minkowski space-time
\begin{equation}
    | \Phi^{+} \rangle = \frac{1}{\sqrt{2}} \left( |0_{A} \rangle_{M} \otimes |0_{B} \rangle_{M} 
+ |1_{A} \rangle_{M} \otimes |1_{B} \rangle_{M} \right)
\label{Bellstatephiplus}                                                    
\end{equation}
Assuming relativistic motion as noisy dynamics we have
\begin{equation}
    \rho^{Q} \rightarrow \rho^{Q^{\prime}} = \mathcal{E}^{Q}(\rho^{Q}) = \sum_{n} A_{n}^{Q} \rho^{Q} {A_{n}^{Q} }^\dagger
\label{dmtransformation}    
\end{equation}
Here we use the $A_{n}^{Q}$ operator which is Eq. (11) in Ref. \cite{ahn2018unruh} and is also given below:
\begin{equation}
    A_{n}^{Q} = \frac{1}{\sqrt{n!}} \frac{\tanh^{n} r}{\cosh^{2} r}     (\cosh r)^{\hat{n}_{A}} \otimes ( b^{\dagger}_{I} )^{n}
\end{equation}
where $\hat{n}_{A} = a_{A}^{\dagger} a_{A}$ is the number operator acting on the Hilbert space corresponding to Alice \cite{ahn2018unruh}.  
Applying this to the density matrix corresponding to the $|\Phi^{+} \rangle$ state we obtain
\begin{eqnarray}
    \rho^{\prime} &=& \sum_{n =0}^{\infty} A_{n}^{Q} \rho^{Q}  {A_{n}^{Q} }^\dagger  \nonumber \\
                  &=& \frac{1}{2} \sum_{n=0}^{\infty} A_{n}^{Q} \bigg( |00 \rangle \langle 00 | + |00 \rangle \langle 11| 
                                   + | 11 \rangle \langle 00 | \nonumber \\
                  & &  \phantom{\frac{1}{2} \sum_{n=0}^{\infty} A_{n}^{Q} } + |11 \rangle \langle 11| \bigg)  {A_{n}^{Q} }^\dagger  .
\end{eqnarray}
The action of the operator $A_{n}^{Q}$ on the $|00 \rangle$ and $|11 \rangle$ gives
\begin{eqnarray}
 A_{n}^{Q} |00 \rangle  &=& \frac{\tanh^{n}r}{\cosh^{2}r}  |0 \rangle |n\rangle_{\text{I}}  |n\rangle_{\text{II}}    \\
 A_{n}^{Q} |11 \rangle  &=& \frac{\tanh^{n}r}{\cosh r} \sqrt{n+1} |1 \rangle |n+1 \rangle_{\text{I}} |n\rangle_{\text{II}} 
\end{eqnarray}
Implementing the action of the operator $A_{n}^{Q}$ and tracing out the modes in the Rindler II region we get
\begin{eqnarray}
   \rho^{\prime} &=&  \frac{1}{2} \Bigg[ \sum_{n=0}^{\infty} \frac{\tanh^{2n}r}{\cosh^{4}r} |0 \rangle |n\rangle_{\text{I}} \langle 0| \langle n|_{\text{I}} 
                      \nonumber \\
                 & &  + \frac{\tanh^{2n}r}{\cosh^{3}r} \sqrt{n+1} |0 \rangle |n\rangle_{\text{I}}  \langle 1| \langle n+1|_{\text{I}}    \nonumber \\
                 & &  + \frac{\tanh^{2n}r}{\cosh^{3}r} \sqrt{n+1} |1 \rangle |n+1\rangle_{\text{I}} \langle 0| \langle n|_{\text{I}}    \nonumber \\
                 & &  + \frac{\tanh^{2n}r}{\cosh^{2}r} (n+1) |1 \rangle |n+1\rangle_{\text{I}} \langle 1| \langle n+1|_{\text{I}} 
\label{phiplustransformed}                 
\end{eqnarray}
To verify the trace preserving property we need to compare the trace of the initial state and the final state.  It is well known that the 
initial state $|\Phi^{+}\rangle$ has a trace of $1$. Now if we compute the trace of the state in Eq. (\ref{phiplustransformed}) we get
\begin{equation}
    {\rm Tr} (\rho^{\prime}) =  \frac{1}{2} \left[ \sum_{n=0}^{\infty} \frac{\tanh^{2n}r}{\cosh^{4}r} 
                                + \sum_{n=0}^{\infty} \frac{\tanh^{2n}r}{\cosh^{2}r} (n+1) \right]
    \label{traceofrhoprime}
\end{equation}
From Ref. \cite{ahn2018unruh} we know the following summations
\begin{eqnarray}
    \sum_{n=0}^{\infty}  \tanh^{2n}r &=&  \cosh^{2}r \\
    \sum_{n=0}^{\infty} (n+1) \tanh^{2n}r  &=& \cosh^{4}r 
    \label{summations} 
\end{eqnarray}
Substituting the summations in Eq. (\ref{summations}) in Eq. (\ref{traceofrhoprime}) we get 
\begin{equation}
 {\rm Tr} (\rho^{\prime}) =  \frac{1}{2} \left[ \frac{1}{\cosh^{2}r} + \cosh^{2} r \right] \neq 1 .
\end{equation}
Thus we prove that the operator $A_{n}^{Q}$ is not trace preserving with respect to the Bell state $|\Phi^{+} \rangle$. 
Hence the operator given in Ref. \cite{ahn2018unruh} is not universal in the sense that it does not work for all states.

\end{document}